\documentclass[12pt]{article}
\usepackage{amssymb}
\usepackage{amsmath}
\usepackage{amsthm}
\usepackage{xcolor}
\newtheorem{thm}{Theorem}
\newtheorem{lem}{Lemma}

\usepackage{amssymb,amsmath}
\usepackage{a4wide}
\usepackage{graphicx}
\usepackage{setspace}
\usepackage{natbib}
\usepackage{multicol}
\def\i{\mathrm{i}}

\usepackage{enumerate}

\title{A Frequency-Domain NonStationarity Test for  dependent data}
\author{ M. Ould Haye$^1$\footnote{corresponding author : {mohamedouhaye@cunet.carleton.ca}} \and
  A. Philippe$^2$
}
\date{$^1$\small  School of Mathematics and Statistics. \\
Carleton University, 1125 Colonel By Dr. Ottawa, ON, Canada, K1S 5B6
\\ $^2$  Laboratoire de Math\'{e}matiques Jean Leray, UMR6629 CNRS \\
2 rue de la Houssiniere, Nantes Université, 44 322 Nantes France.}

\begin{document}
\maketitle
\begin{abstract}
Distinguishing long-memory behaviour from nonstationarity is challenging, as both produce slowly decaying sample autocovariances. Existing stationarity tests either fail to account for long-memory processes or exhibit poor empirical size, particularly near the boundary between stationarity and nonstationarity. We propose a new, parameter-free testing procedure based on the evaluation of periodograms across multiple epochs. The limiting distributions derived here are obtained under stationarity and nonstationarity assumptions and analytically tractable, expressed as finite sums of weighted independent $\chi^2$ random variables. Simulation studies indicate that the proposed method performs favorably compared to existing approaches.
\end{abstract}

\section{Introduction}
Distinguishing long--memory dependence from genuine nonstationarity is a longstanding challenge, since both generate slowly decaying sample autocovariances and similar low--frequency behaviour. Classical unit--root and stationarity tests, such as those of \cite{pptest}  and \citet{MR3505787}, are primarily designed for weakly dependent or short--memory processes. The $V/S$ test of \cite{MR2328526} was the first to accommodate long memory within a stationarity framework, but it tends to be oversized when the memory parameter $d$ approaches $1/2$, and thus loses reliability near the stationarity boundary where dependence is strongest.

\medskip
Both the Dickey--Fuller and the $V/S$ tests are \emph{time--domain} procedures. Their statistics are based on partial--sum or cumulative representations of the series, leading to limiting distributions expressed as stochastic integrals of Brownian or fractional Brownian motions. These stochastic--integral limits are analytically intractable, and practical implementation requires simulation or pre--computed critical values.

\medskip
In this work we propose a \emph{frequency--domain Dickey--Fuller--type test} that remains valid under strong dependence while yielding an explicit, tractable limit law. The test discriminates between an integrated process of order one $I(1)$ and a stationary process $I(0)$: that is, with $d$ being the memory parameter of the data generating process (DGP) $X_t$,  we consider the one-sided hypothesis test
\[
H_0:\; 3/2>d \ge \tfrac12  (X_t\text{ is nonstationary, }I(1)), 
\quad  
H_1:\; -1/2<d < \tfrac12  (X_t\text{ is stationary, }I(0)).
\]
The nonstationarity cutoff value $d=1/2$ corresponds to the so-called $1/f$ noise model and marks the border between I(1) and I(0) regimes, \cite{Robinson} refers to this as "just nonstationary regime".
Our framework covers both regimes with memory parameters in the full range $d\in(-1/2,\,3/2)$ and does not rely on any \emph{semi--parametric assumption} on the spectral density such as $f(\lambda)=|\lambda|^{-2d}f^*(\lambda)$, with $f^*$  continuous, bounded, or slowly varying. \\
The key theoretical result, Theorem 1, establishes the limiting distribution of a properly normalized periodogram under both $I(0)$ and $I(1)$ settings. The limit is expressed as a finite weighted sum of independent $\chi^2(1)$ random variables, with weights determined by the memory parameter through an explicit positive-definite covariance matrix. This tractable representation stands in sharp contrast to the stochastic-integral limits typically arising in time-domain Dickey–Fuller tests. Theorem 1 is of broader interest, as it unifies the asymptotic behaviour of the normalized periodogram for stationary and integrated processes, thereby providing the theoretical foundation for a new frequency-domain stationarity test - analogous to the time-domain Dickey–Fuller test, yet it extends well beyond the traditional AR($p$) versus unit-root framework by accommodating  tests of short- and long-range dependence against general $I(1)$ nonstationarity.\\
We develop a nonstationarity parameter-free statistical test that does not require estimating the memory parameter $d$. Unlike the  V/S test, our test detects the stationarity for all $d$ in (-1/2,1/2) whereas the  V/S test requires $d$ to be in a compact $[a,b]\subset(-1/2,1/2)$ and a plug-in estimator has to be used. Another limitation in V/S consists in it  being oversized, with its empirical size significantly much larger than the nominal level, for $d\ge.4$, and therefore one cannot choose $b$ too close to 1/2. In contrast, the proposed test does not suffer form  this issue maintaining  empirical sizes close to the nominal level even for $d$  close to 1/2. \\\\
From a sample $X_1,\ldots,X_n$ of the process $(X_t)$, we are
interested in building a testing procedure to discriminate between dependence and non-stationarity. \\
The proposed statistic is built from periodograms taken at different epochs. More precisely, the procedure is as follows:
we split our initial sample $X_1,\ldots,X_n$ into $m$ blocks (or
epochs), each of size $\ell$, and  we construct  the periodogram
$I_{n,i}$ on the $i$th block $X_{(i-1) \ell+1},\ldots,X_{i\ell }$.
Let
\begin{equation}\label{final}
Q_{n,m}(s,d)=\left(m^{-2d}\right)\sum_{j=1}^s \frac{I_n(\lambda_j)}{\frac{1}{m}\sum_{h=1}^mI_{n,h}(\lambda'_j)},
\end{equation}
where $\lambda_j=2\pi j/n$, and  $\lambda'_j=2\pi j/\ell$, $j=1,\ldots,s$, for fixed $s$, are the Fourier frequencies, and 
\begin{equation}\label{perio}
I_n(\lambda_j)=\frac{1}{2\pi n}\left\vert\sum_{t=1}^nX_te^{\i t\lambda_j}\right\vert^2,\qquad I_{n,h}(\lambda_j')=\frac{1}{2\pi \ell}\left\vert\sum_{t=(h-1)\ell+1}^{h\ell}X_te^{\i t\lambda_j'}\right\vert^2.
\end{equation}
The statistic $Q_{n,m}(s,d)$ can be viewed as sum of self normalized periodograms $I_n(\lambda_j)$, $j=1,\ldots,s$.\\
In our 2018 paper (\cite{Gromykov}), we introduced the statistic $Q_{n,m}(s,0)$ to test short-memory versus long-memory behaviour and established its asymptotic distribution under the short-memory assumption $d=0$. Extending these results to the long or anticipative memory range ($-1/2<d<1/2$) and to the nonstationary regime ($1/2\le d<3/2$) requires a fundamentally new approach and, in particular, the identification of an appropriate normalization ensuring a nondegenerate limit. This step is crucial not only in the nonstationary $I(1)$ framework---where the spectral density does not exist and standard frequency-domain arguments fail---but also within the stationary regime, where long or anticipative dependence alters both the scaling and the very nature of the limit, as the latter inherits the dependence structure of the process. The key modification consists in introducing the normalizing factor $m^{-2d}$, as in (\ref{final}), which enables a unified asymptotic treatment of both stationary and integrated processes. This adjustment yields distinct limiting distributions reflecting the effects of long or anticipative dependence in the stationary $I(0)$ case and of integration in the nonstationary $I(1)$ setting, thereby fundamentally shaping the asymptotic properties of the properly normalized periodograms in (\ref{perio}).
  
We assume that $m=m(n),\textrm{ and } \ell=\ell(n)\to\infty$ as $n\to\infty$. It is  important to emphasize the fact that $m$ and $\ell$ increase with $n$ and are not constant, and that $m=n/\ell$.
We are simply using notation $m$ and $\ell$ rather than $m(n)$ and
$\ell(n)$ only for the sake of simplicity. 
The rest of the paper is organized as follows. 
Section \ref{sec:2} contains main limiting theorems related to the statistic $Q_{n,m}(s,d)$, including its behaviour under a wide range of nonstationary alternatives. 
In Section \ref{simulation} we propose a  test statistic to detect nonstationarity  based on the asymptotic properties of the statistic $Q_{n,m}(s,d)$ obtained in  Section \ref{sec:2} and we present a Carlo study  to illustrate the performance of our proposed method.  
Section \ref{proofs} contains   proofs of the lemmas 1-4.
\section{Main Results}\label{sec:2}

We say that a stochastic process \((X_t)_{t \ge 0}\) is an \(I(0)\) process with memory parameter \(d \in [-\tfrac{1}{2}, \tfrac{1}{2})\) if it can be written as  
\begin{equation}\label{linear}
X_t = \sum_{j=0}^\infty a_j\, \epsilon_{t-j},
\end{equation}
where \((\epsilon_t)\) are independent and identically distributed random variables with zero mean, variance \(\sigma_\epsilon^2 = \mathbb{E}(\epsilon_1^2)\), and finite fourth moment \(\eta = \mathbb{E}(\epsilon_1^4)\). 
The coefficients \((a_j)\) determine the degree of memory through their asymptotic behavior:
\begin{equation}\label{summable}
\sum_{j=0}^\infty |a_j| < \infty, \quad \sum_{j=0}^\infty a_j\neq0,\qquad\text{when } d = 0,
\end{equation}
\begin{equation}\label{positived}
a_j \sim c(d)\, j^{-1+d}, \quad \text{for } 0 < d < \tfrac{1}{2},
\end{equation}
\begin{equation}\label{negatived1}
\sum_{j=0}^\infty a_j = 0, \qquad 
a_j = c(d)\, j^{-1+d}\,(1 + O(j^{-1})), \quad \text{for } -\tfrac{1}{2} \le d < 0,
\end{equation}
where \(c(d) > 0\).  

We say that \((X_t)\) is an integrated process of order one, or \(I(1)\), with memory parameter \(1/2 \le d < 3/2\), if its first-differenced process
\[
Y_t = X_t - X_{t-1}
\]
is an \(I(0)\) process with memory parameter \(d - 1\).

Condition (\ref{summable}), imposed when $d=0$, implies that 
 $$
 \sum_{h=-\infty}^\infty\vert\gamma(h)\vert<\infty.
 $$
 According to Proposition 3.2.1. of \cite{MR2977317}, when $d\neq 0$ the conditions \eqref{positived} and   (\ref{negatived1}) imply that the 
 \begin{equation}\label{gammad}
\gamma(h)\sim C(d) h^{2d-1},\qquad\textrm{as }h\to\infty,\qquad C(d)=\sigma^2_\epsilon c(d)B(d,1-2d),
 \end{equation}
 where $B$ is the $\beta$ function, and 
 $a_n\sim b_n$ means that $a_n/b_n\to1$ as $n\to\infty$.

In what  follows, we will focus on the statistic $Q_{n,m}(s,d)$ defined in (\ref{final}). In the next theorem  we give the  asymptotic
distribution of $Q_{n,m}(s,d)$ for I(1) and (0) processes. Let
\begin{equation}\label{const delta}
\delta(u)=\begin{cases}8\sigma^2_\epsilon c(-1/2)&\textrm{if }u=-1/2,\\\\
   \frac{C(u)}{u(2u+1)}&\textrm{if }-1/2<u<1/2.
\end{cases}
\end{equation}
 \begin{thm}\label{Th1}
Let $(X_t)$ be an $I(0)$ or $I(1)$ stochastic process with memory parameter $d \in (-1/2,\,3/2)$. Then,
\begin{equation}\label{d limit}
Q_{n,m}(s,d) \;\overset{\mathcal{D}}{\longrightarrow}\;
Q(s,d) = \sum_{i=1}^{2s} \zeta_i(d)\, Q_i, 
\end{equation}
where $\{Q_i\}$ are i.i.d.\ $\chi^2(1)$ random variables, and $\{\zeta_i(d)\}$ are the eigenvalues of the  matrix
$\Sigma(d)D^{-1}$, where $\Sigma(d)$ is the block matrix
\begin{equation}\label{Sigma}
\Sigma(d) =
\begin{pmatrix}
\Sigma^{(c)} & 0\\[3pt]
0 &\Sigma^{(s)}
\end{pmatrix},
\end{equation}
and $D$ is the diagonal matrix with diagonal entries
$$D_{ii} =\left(\Sigma^{(c)}_{ii}(d) + \Sigma^{(s)}_{ii}(d)\right),\qquad i=1,\ldots,s,$$
and for $i,j=1,\ldots,s$, 
\[
\Sigma^{(c)}_{ij}(d) =
\begin{cases}
\displaystyle \frac{\delta(-1/2)}{2}\int_{[0,1]^2}\!\cos(2\pi i x)\cos(2\pi j y)\,
\bigl(-\log|x-y|\bigr)\,dx\,dy,
& \text{if } d=\tfrac{1}{2},\\[10pt]
\displaystyle\frac{\delta(d-1)}{2}\int_{[0,1]^2}\!\cos(2\pi i x)\cos(2\pi j y)\,
|x-y|^{2d-1}\,dx\,dy,
& \text{if } 1/2<d<3/2,\\[10pt]
-\delta(d)\Bigg[a(d)+
2\pi^2 i j\displaystyle\int_{[0,1]^2}\sin(2\pi ix)\sin(2\pi jy)\vert x- y\vert^{2d+1}dxdy\Bigg] 
& \text{if } -1/2<d<1/20,
\end{cases}
\]
where
$$
a(d)=1-((2d+1)\int_0^1x^{2d}\left(\cos(2\pi ix)+\cos(2\pi jx)\right)dx,
$$
and similarly,
\[
\Sigma^{(s)}_{ij}(d) =
\begin{cases}
\displaystyle\frac{\delta(-1/2)}{2} \int_{[0,1]^2}\!\sin(2\pi i x)\sin(2\pi j y)\,
\bigl(-\log|x-y|\bigr)\,dx\,dy,
& \text{if } d=\tfrac{1}{2},\\[10pt]
\frac{\delta(d-1)}{2}\displaystyle\int_{[0,1]^2}\!\sin(2\pi i x)\sin(2\pi j y)\,
|x-y|^{2d-1}\,dx\,dy,
& \text{if } 1/2<d<3/2,\\[10pt]
-\delta(d)(2\pi^2 ij)\displaystyle\int_{[0,1]^2}\!\cos(2\pi i x)\cos(2\pi j y)\,
|x-y|^{2d+1}\,dx\,dy,
& \text{if } -1/2<d<1/2.
\end{cases}
\]
\end{thm}
\begin{proof}

Let $\lambda_1,\ldots,\lambda_s$ be the $s$ first  Fourier frequencies. 
With $Y_k=X_k-X_{k-1}$ being the differenced process, we have
$$
X_k=\sum_{k=1}^kY_t+X_0.
$$
Since, for any positive integer $u,s$ and fixed $j=1,\ldots,s$
\begin{equation}\label{sumzero2}
\sum_{t=1}^ue^{2\pi \i tj/u}=0,
\end{equation}
the periodograms built from $(X_k)_{1\le k\le n}$ and $(X_k-X_0)_{1\le i\le n}$ are the same and hence the statistics $Q_{n,m}(s,d)$ built from $X_k$ and $S_k(Y):=Y_1+\cdots+Y_k$ are equal, and therefore
convergence (\ref{d limit}) will follow from establishing the following lemmas.
\end{proof}
\begin{lem}\label{invertible}
The matrices $\Sigma^{(c)}$ and $\Sigma^{(s)}$ are positive definite, and hence $\Sigma(d)$ is positive definite as well.
\end{lem}
\begin{lem} \label{unit root} Let $X_t$ be an I(1) with $d=1/2$. Then the differenced process $Y$ satisfies the following:  As $n\to\infty$,
   $$
   \textrm{Var}(S_n(Y))\sim \delta(-1/2)\log n.
   $$ 
\end{lem}
\begin{lem}\label{finite dim} 
(1)  if $1/2\le d<3/2$,
\begin{eqnarray*}\lefteqn{
\Bigg[\left(\frac{1}{n^{1/2+d}}\sum_{k=1}^n\cos(k\lambda_{1})S_k(Y)\right),
\cdots,\left(\frac{1}{n^{1/2+d}}\sum_{k=1}^n\cos(k\lambda_{s})S_k(Y)\right),}\\
&& \left(\frac{1}{n^{1/2+d}}\sum_{k=1}^n\sin(k\lambda_1)S_k(Y)\right),\cdots,\left(\frac{1}{n^{1/2+d}}\sum_{k=1}^n\sin(k\lambda_{s})S_k(Y)\right)\Bigg]\overset{d}{\rightarrow}\mathcal{N}\left((0,\Sigma(d)\right), \nonumber
\end{eqnarray*}
(2) if $-1/2<d<1/2$
\begin{eqnarray*}\lefteqn{
\Bigg[\left(\frac{1}{n^{1/2+d}}\sum_{k=1}^n\cos(k\lambda_{1})X_k\right),
\cdots,\left(\frac{1}{n^{1/2+d}}\sum_{k=1}^n\cos(k\lambda_{s})X_k\right),}\nonumber\\
&& \left(\frac{1}{n^{1/2+d}}\sum_{k=1}^n\sin(k\lambda_1)X_k\right),\cdots,\left(\frac{1}{n^{1/2+d}}\sum_{k=1}^n\sin(k\lambda_{s})X_k\right)\Bigg]\overset{d}{\rightarrow}\mathcal{N}\left((0,\Sigma(d)\right). 
\end{eqnarray*}
\end{lem}
\begin{lem} \label{denom} For any $j=1,\ldots,j$, $-1/2<d<3/2$, and as $n\to\infty$,
\begin{equation}\label{mena limit}
\mathbb{E}\left[\left(\frac{1}{m}\sum_{h=1}^m\frac{I_{n,h}(\lambda'_j)}{\ell^{2d}}-D_{jj}\right)^2\right]\to 0.
\end{equation}
\end{lem}
\section{Testing Procedure and Monte Carlo Study }\label{simulation}

\subsection{Definition of the critical region }

Testing stationarity is formulated as the one-sided hypothesis  test
\[
H_0: I(1) \quad \text{versus} \quad H_1: I(0),
\]
where the test statistic is $Q_{n,m}(s,1/2)$ and the limiting distribution is $Q(s,1/2)$.  As shown, this procedure is completely parameter free, meaning that it does not require estimation of the memory parameter $d$. We reject $H_0$ in favour of stationarity whenever  
\[
Q_{n,m}(s,1/2) < q_{\alpha},
\]
where $q_{\alpha}$ denotes the $5\%$ quantile of the limiting distribution $Q(s,1/2)$. Otherwise, we fail to reject $H_0$.\\

This test has asymptotical level $\alpha$ : 
$$ \sup_{I(1)} P(Q_{n,m}(s,1/2) < q_{\alpha})  =  \alpha.$$ 
Indeed the supremum over $I(1)$ is attained  at $d=1/2$ according to Theorem \ref{Th1}. 
Moreover under $I(0)$ the power function tends to 1 as $n$ goes to infinity. 

\subsection{Superimposed Cumulative Distribution Functions (CDF)}
We illustrate the convergence result \eqref{d limit} obtained in Theorem \ref{Th1} for different values of $d$. In particular the value $d=1/2$ showing that the empirical size is very close to the nominal level. 
Figure~\ref{fig:cdf} compares finite-sample and limiting CDFs
for $Q_{n,m}(s=2,d)$ under different $d$ values.

\subsection{Empirical Size and Power} 
Figure \ref{fig:power_panel} represents the empirical size and power analysis of the statistic $Q_{n,m}(s)$ under two different DGPs: FARIMA(0, $d$, 0) and AR(1). It displays the size and power as functions of the memory parameter $d$ and the autoregressive coefficient $\phi$ under the respective stationarity assumptions. In particular, it shows that while the empirical size remains close to the nominal level (here $0.05$), the power increases sharply toward one as the process moves below the nonstationarity boundary--$d = 1/2$ for the FARIMA(0, $d$, 0) model and $\phi = 1$ for the AR(1) model.  We scanned several values of $s$, number of Fourier frequencies to be included in the test statistic and found that $s=2$ is a balanced trade-off choice between empirical sizes and powers.\\
\textbf{Comparison with the V/S test.} The fact that $d$ be in a compact of (-1/2,1/2) and estimated in V/S may explain its high emprical size.

\begin{table}[h!]
\centering
\caption{Empirical size of the V/S test at $\alpha=0.05$ (sample sizes $n=1000$ and 4000).}
\begin{tabular}{lcc}
\hline
\textbf{DGP} & \textbf{Parameters} & \textbf{V/S empirical size (range)} \\
\hline
FARIMA$(0,d,0)$ & $d=0.45$ & $0.15$--$0.20$ \\
AR$(1)$ & $\phi=0.8$ & $0.12$--$0.18$ \\
\hline
\end{tabular}

\medskip
\raggedright\footnotesize
\end{table}
 Our procedure is a one-sided test of nonstationarity vs.\ stationarity. 
By contrast, the V/S test stationarity versus nonstationarity, i.e. asks whether there exists some $d<\tfrac{1}{2}$ consistent with the DGP and requires plug-in estimation of $d$.
\newpage 
\begin{figure}[h!]
  \centering
  \includegraphics[width=0.31\textwidth,height = .3\textheight]{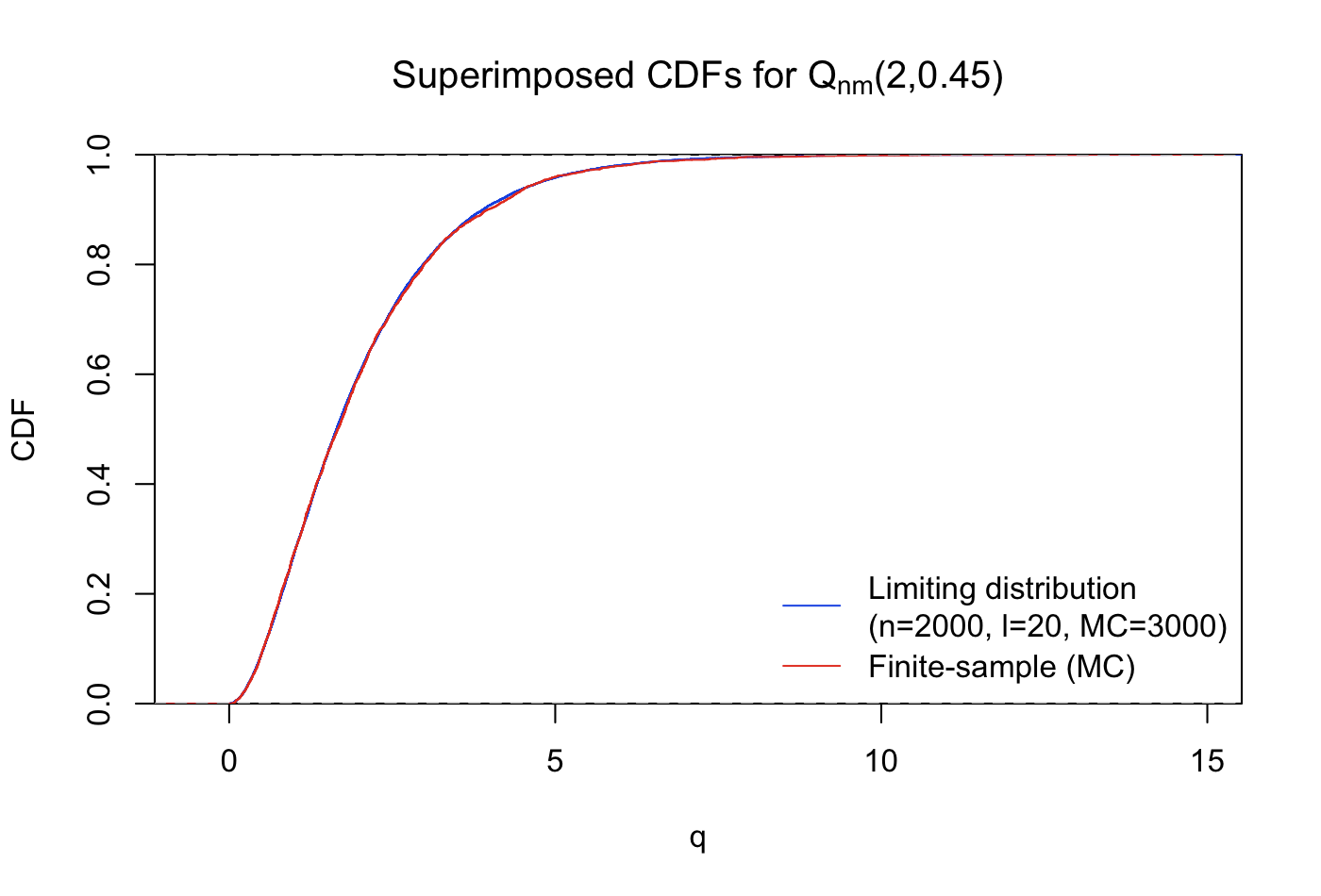}
  \includegraphics[width=0.31\textwidth,height = .3\textheight]{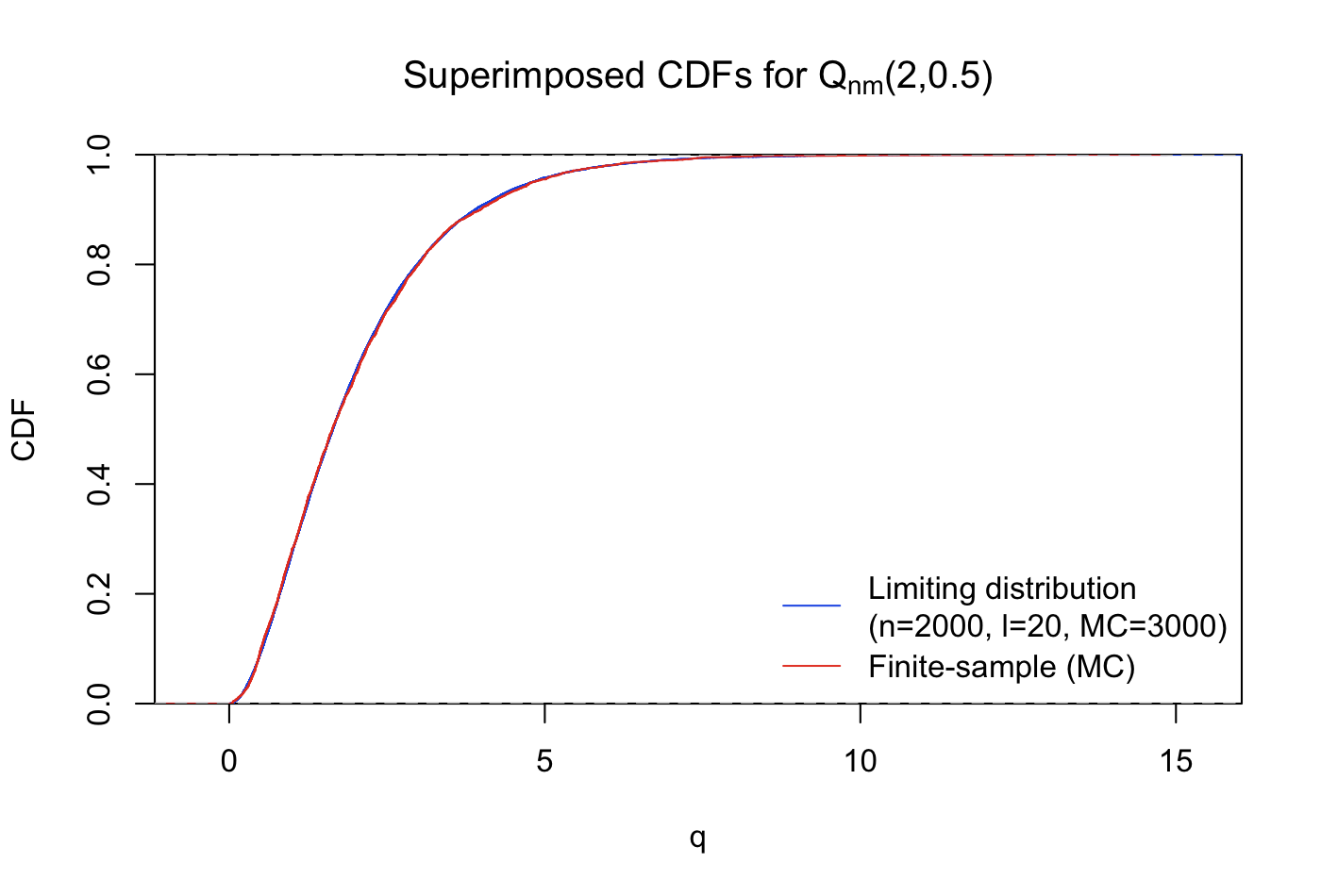}
  \includegraphics[width=0.31\textwidth,height = .3\textheight]{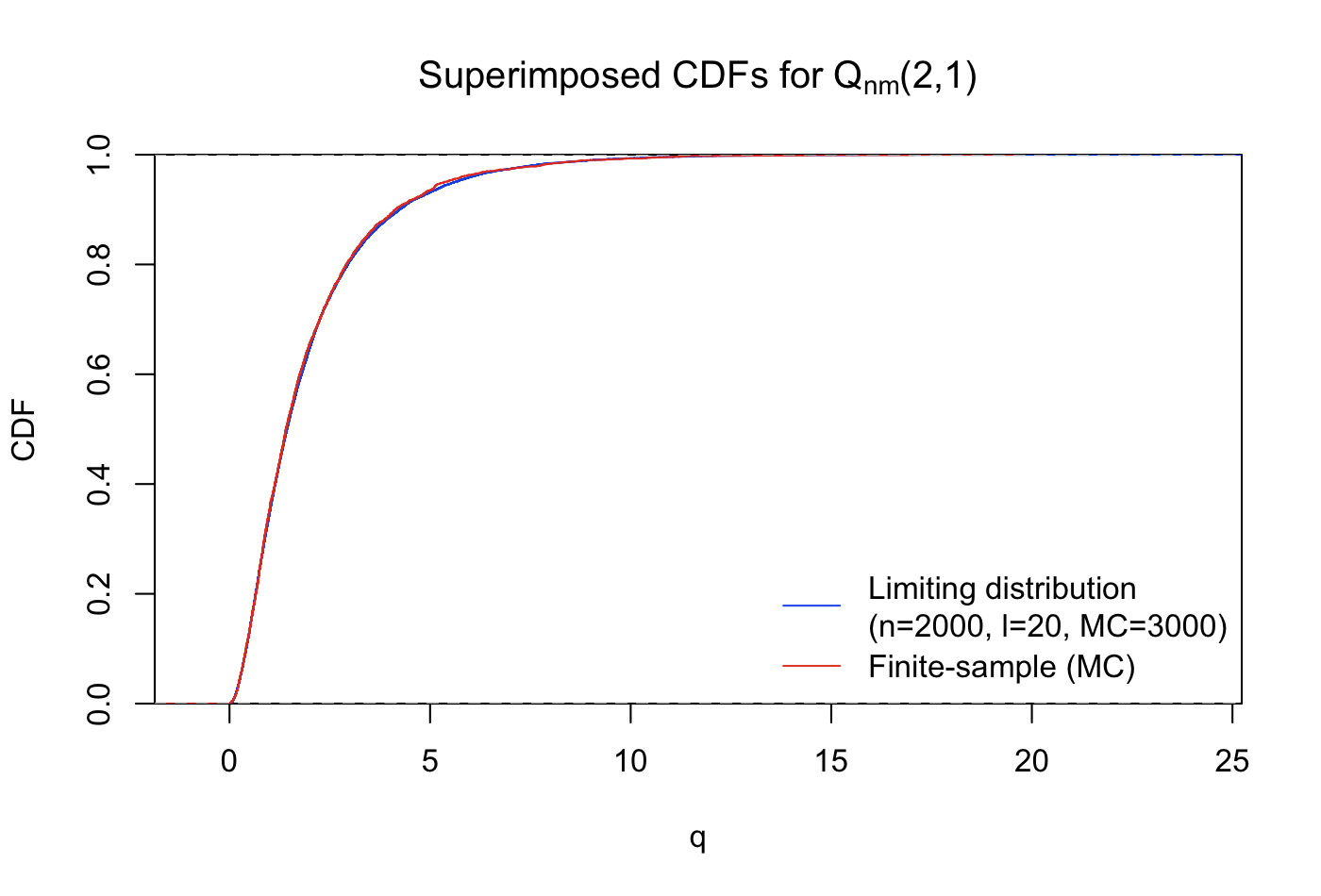}
  \caption{Superimposed CDFs of $Q_{n,m}(2,d)$:
  finite-sample (red) vs.\ limiting distribution (blue)
  for $d=0.45$, $0.5$, and $1.0$ ($n=2000$, $\ell=10$, $m=200$).}
  \label{fig:cdf}
\end{figure}

\begin{figure}[h!]
    \centering
    \includegraphics[width=0.9\textwidth]{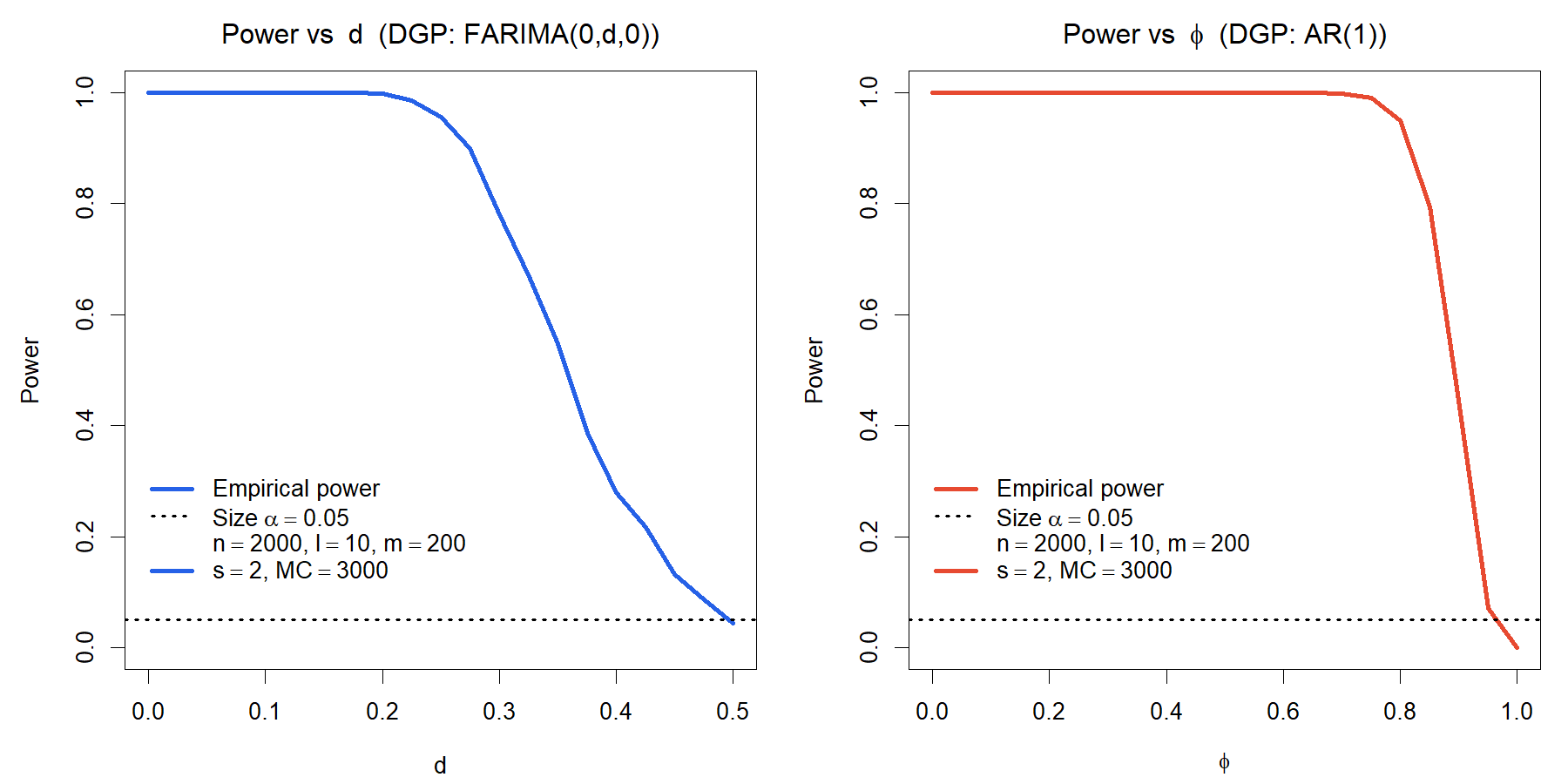}
    \caption{Empirical power functions of the statistic $Q_{n,m}(s)$ for two DGPs:
    (Left) FARIMA$(0,d,0)$ with $d \leq \tfrac{1}{2}$ and
    (Right) AR$(1)$ with $\phi \in [0,1]$.
    Simulations used $n=2000$, $\ell=10$, $m=200$, $s=2$, and $3000$ Monte Carlo replications.}
    \label{fig:power_panel}
\end{figure}
\clearpage
\section{Proofs of Lemmas}\label{proofs}
{\bf Proof of Lemma \ref{invertible}}
\begin{proof}
We consider three cases according to the value of $d$.

\medskip
\noindent
\textbf{Case 1:} $\mathbf{d=\tfrac{1}{2}}$.

Let $a=(a_1,\ldots,a_s)\in\mathbb{R}^s\setminus\{0\}$ and define
\[
\Phi(x) = \sum_{j=1}^{s} a_j \cos(2\pi jx), \qquad x\in(0,1),
\]
extended by $0$ outside $(0,1)$. Clearly, $\displaystyle \int_{\mathbb{R}}\Phi(x)\,dx=0$.  
Using the integral representation
\[
-\log r = \int_0^\infty \frac{e^{-rs}-e^{-s}}{s}\,ds,
\]
we have
\[
a'\Sigma^{(c)}a = \iint_{\mathbb{R}^2} \Phi(x)\Phi(y)
\!\int_0^\infty \frac{e^{-s|x-y|}-e^{-s}}{s}\,ds\,dx\,dy.
\]
Interchanging the order of integration gives
\[
a'\Sigma^{(c)}a = \int_0^\infty \frac{1}{s}
\Bigl\langle \Phi,\, \Phi*K_s \Bigr\rangle\,ds,
\]
where $K_s(t)=e^{-s|t|}$ and $\widehat{K_s}(\xi)=\dfrac{2s}{s^2+\xi^2}$.  
By the Plancherel identity and the convolution theorem,
\[
\langle \Phi,\,\Phi*K_s\rangle
= \frac{1}{2\pi}\int_{\mathbb{R}} |\widehat{\Phi}(\xi)|^2 \,\widehat{K_s}(\xi)\,d\xi.
\]
Hence,
\[
a'\Sigma^{(c)}a
= \frac{1}{\pi}\int_{\mathbb{R}} |\widehat{\Phi}(\xi)|^2
\Bigl(\int_0^\infty \frac{ds}{s^2+\xi^2}\Bigr)\,d\xi
= \frac{1}{2}\int_{\mathbb{R}}\frac{|\widehat{\Phi}(\xi)|^2}{|\xi|}\,d\xi.
\]
The integral is finite because $\Phi$ has compact support and
\[
\widehat{\Phi}(0)=\int_0^1\Phi(x)\,dx=0,
\]
so that $\widehat{\Phi}$ is continuously differentiable with
$|\widehat{\Phi}(\xi)|\le C|\xi|$ near $0$ and $|\widehat{\Phi}(\xi)|=O(|\xi|^{-2})$ as $|\xi|\to\infty$.  
Since $\Phi\not\equiv0$, the integral is strictly positive, proving positive definiteness.  
The same argument applies to $\Sigma^{(s)}$.

\medskip
\noindent
\textbf{Case 2:} $\mathbf{\tfrac{1}{2}\le d<\tfrac{3}{2}}$.

Let $\beta=2d-1\in(0,2)$.  Using  (\cite{gradshteyn2007tables} 3.823) and the fact that for any non integer value $z$,
\[
\Gamma(-z) = -\,\frac{\pi}{z\,\Gamma(z)\,\sin(\pi z)}.
\]
(the case $\beta=1$ is obtained directly as $\Gamma$  function is not defined for nonpositive integers), we obtain
\[
|u|^{\beta}
= c_\beta \int_0^\infty (1-\cos(2\pi us))\,s^{-1-\beta}\,ds,
\quad
c_\beta=\frac{2^{1-\beta}\Gamma(1+\beta)\sin(\pi\beta/2)}{\pi^{1+\beta}}>0,
\]
and the fact that $\int_0^1\Phi(x)\,dx=0$, we get
\[
a'\Sigma^{(c)}a
= c_\beta \int_0^\infty s^{-1-\beta}\,|\widehat{\Phi}(s)|^2\,ds > 0.
\]

\medskip
\noindent
\textbf{Case 3:} $\mathbf{-\tfrac{1}{2}<d<\tfrac{1}{2}}$.

Here $\beta=2d+1\in(0,2)$, and we can go back to case 2.
The same reasoning applies to $\Sigma^{(s)}$.

\smallskip
Thus, in all cases, $\Sigma^{(c)}$ and $\Sigma^{(s)}$ are positive definite, and therefore $\Sigma(d)$ is positive definite.
\end{proof}
{\bf Proof of Lemma \ref{unit root}}
\begin{proof}
First, we start by noticing that the proof of Proposition 3.2.1 of \cite{MR2977317} can be adapted to include the case $d=-1/2$ to obtain that the covariance function of the differenced process $Y_k$ (which is an I(0) process with $d=-1/2$) satisfies
$$
\sum_{k=-\infty}^\infty\gamma_Y(k)=0
$$
and has the form
$$
\gamma_Y(k)=k^{-2}h(k),
$$
where $h$ is a slowly varying function at infinity and $h(k)\to c:=4\sigma^2_\epsilon c(-1/2)$.
We have
\begin{eqnarray}
\textrm{Var}(S_n(Y))&=&n\left(\gamma_Y(0)+2\sum_{k=1}^n\left(1-\frac{k}{n}\right)\gamma_Y(k)\right)\nonumber\\
&=&n\left(\sum_{k=-n}^n\gamma_Y(k)\right)-2\sum_{k=1}^nk\gamma_Y(k)\nonumber\\
&=&2n\sum_{k=n+1}^\infty k^{-2}h(k)-2\sum_{k=1}^nk^{-1}h(k)\label{karamata}\\
&\sim&2(4\sigma^2c(-1/2))=\delta(-1/2)\log(n),\qquad\textrm{as }n\to\infty,\nonumber
\end{eqnarray}
since $h(n)=\delta(-1/2)\left(1+o(1)\right):=\delta(1/2)+h_1(n),$ where $h_1(n)\to0$ as $n\to\infty$.  We used Karamata Theorem for the first sum in  (\ref{karamata}) to say that it is equivalent to $n^{-1}h(n)$ and hence the first term  is converging to a positive constant. The second sum in (\ref{karamata}) is equivalent to
$$
\delta(1/2)\int_1^nx^{-1}dx+\int_1^nx^{-1}h_1(x)dx=\delta(1/2)\log n+\int_1^nx^{-1}h_1(x)dx.
$$
For  certain positive constant $A$, we have
\begin{eqnarray*}
    \int_1^nx^{-1}h_1(x)dx&=&\int_1^{\log n}x^{-1}h_1(x)dx+\int_{\log n}^nx^{-1}h_1(x)dx\\
    &\le& A\log(\log n)+(\log n)\underset{x\ge \log n}{\sup}h_1(x)=o(\log n).
\end{eqnarray*}
\end{proof} 
{\bf Proof of Lemma \ref{finite dim}}
\begin{proof}
{\bf Case $\mathbf{d=1/2}$}
Applying Lemma \ref{unit root}, and using the fact that $ab=1/2\left(a^2+b^2-(a-b)^2\right)$, the stationarity of the increments of $S_k(Y)$ and (\ref{sumzero2}) we obtain that for every fixed $i,j=1,\ldots,s$ and as $n\to\infty$, 
\begin{eqnarray*}\lefteqn{
    \textrm{Cov}\left[\left(\frac{1}{n}\sum_{k=1}^n\cos(k\lambda_{i})S_k(Y)\right),\left(\frac{1}{n}\sum_{k=1}^n\cos(k\lambda_{j})S_k(Y)\right)\right]}\\
    &&=\frac{1}{n^2}\sum_{k=1}^n\sum_{k'=1}^n\cos\left(\frac{2\pi ik}{n}\right)\cos\left(\frac{2\pi jk}{n}\right)\mathbb{E}\left(S_k(Y)S_{k'}(Y)\right)\\ 
    &&=\frac{1}{2}\frac{1}{n^2}\sum_{k=1}^n\sum_{k'=1}^n\cos\left(\frac{2\pi ik}{n}\right)\cos\left(\frac{2\pi jk}{n}\right)(-\mathbb{E}(S^2_{\vert k-k'\vert}(Y))\\
    &&\sim\frac{\delta(1/2)}{2}\frac{1}{n^2}\sum_{k=1}^n\sum_{k'=1}^n\cos\left(\frac{2\pi ik}{n}\right)\cos\left(\frac{2\pi jk}{n}\right)\left(-\log\left(\frac{\vert k-k'\vert}{n}\right)\right)\\
&&\to\frac{\delta(1/2)}{2}\int_{[0,1]^2]}\cos(2\pi ix)\cos(2\pi jy)(-\log(\vert x- y\vert)dxdy
\end{eqnarray*}
and similarly
\begin{eqnarray*}\lefteqn{
    \textrm{Cov}\left[\left(\frac{1}{n}\sum_{k=1}^n\sin(k\lambda_{i})S_k(Y)\right),\left(\frac{1}{n}\sum_{k=1}^n\sin(k\lambda_{j})S_k(Y)\right)\right]}\\
&&\to\frac{\delta(1/2)}{2}\int_{[0,1]^2]}\sin(2\pi ix)\sin(2\pi jy)(-\log(\vert x- y\vert)dxdy,
\end{eqnarray*}
\begin{eqnarray*}\lefteqn{
    \textrm{Cov}\left[\left(\frac{1}{n}\sum_{k=1}^n\cos(k\lambda_{i})S_k(Y)\right),\left(\frac{1}{n}\sum_{k=1}^n\sin(k\lambda_{j})S_k(Y)\right)\right]}\\
   && \to\frac{\delta(1/2)}{2}\int_{[0,1]^2]}\cos(2\pi ix)\sin(2\pi jy)(-\log(\vert x- y\vert)dxdy=0.
   \end{eqnarray*}
   We can write,   with the convention that $a_h=0$ for $h<0$
\begin{eqnarray}\lefteqn{
\frac{1}{n}\sum_{k=1}^n\cos(k\lambda_{i})S_k(Y)
=\sum_{i=1}^n\left(\frac{1}{n}\sum_{k=1}^{i-1}\cos(k\lambda_j)\right)Y_i}\label{develop}\\
&&=
\sum_{i=1}^n\left(\frac{1}{n}\sum_{k=i}^n\cos(k\lambda_j)\right)
\sum_{u=-\infty}^ia_{i-u}\epsilon_u\nonumber\\
&&=\sum_{u=-\infty}^n\left(\
\frac{1}{n}\sum_{k=1}^n\left(\sum_{i=1}^ka_{i-u}\right)\cos(k\lambda_j)\right)\epsilon_u
:=\sum_{u=-\infty}^nd^{(c)}_{n,u}\epsilon_u.\label{combine}
\end{eqnarray}
For $u\le0$, and since
$$
\sum_{i=0}^\infty a_i=0,
$$
$$
d^{(c)}_{n,u}=\left(\
\frac{1}{n}\sum_{k=1}^n\left(\sum_{i=0}^{k-u}a_i\right)\cos(k\lambda_j)\right)=
-\frac{1}{n}\sum_{k=1}^n\left(\sum_{i=k-u+1}^\infty a_i\right)\cos(k\lambda_j)
$$
and hence, since $d-1=-1/2$,  $\vert a_k\vert\sim c k^{-3/2}$, as $k\to\infty$,
so that
$$
\sum_{i=\ell}^\infty \vert a_i\vert \sim2c\ell^{-1/2},\qquad\textrm{as }\ell\to\infty,
$$
and hence 
$$
\vert d^{(c)}_{n,u}\vert\le 2\vert c\vert n^{-1/2}\to0,\qquad n\to\infty.
$$
Now consider the case $0<u\le n$.
$$
d^{(c)}_{n,u}=\frac{1}{n}\sum_{k=u}^n\left(\sum_{i=0}^{k-u}a_i\right)\cos(k\lambda_j)
$$
Here again, we get
$$
\vert d^{(c)}_{n,u}\vert\le 2\vert c\vert n^{-1/2}\to0\qquad\textrm{as }n\to\infty,
$$
and hence we have shown that 
\begin{equation}\label{uniform u}
   \underset{u\le n}{\sup}\,\,d^{(c)}_{n,u}\to0 \qquad\textrm{as }n\to\infty,
\end{equation}
and of course, as $n\to\infty$,
$$
\sum_{u=-\infty}^n\left(d^{(c)}_{n,u}\right)^2=\mathbb{E}\left(\frac{1}{n}\sum_{k=1}^n\cos(k\lambda_{j})S_k(Y)\right)\to\Sigma^{(c)}_{jj}(1/2)>0.
$$
and similarly we obtain 
$$
d^{(s)}_{n,u}:=\frac{1}{n}\sum_{k=1}^n\left(\sum_{i=1}^ka_{i-u}\right)\sin(k\lambda_j)\to0\quad\textrm{uniformly in }u\textrm{ as }n\to\infty.
$$
and that
$$
\sum_{u=-\infty}^n\left(d^{(s)}_{n,u}\right)^2=\mathbb{E}\left(\frac{1}{n}\sum_{k=1}^n\sin(k\lambda_{j})S_k(Y)\right)\to\Sigma^{(s)}_{jj}(1/2)>0.
$$
Since $\Sigma(1/2)$ is invertible by virtue of Lemma \ref{invertible} and reasoning similarly as in the proof of  Theorem 4.3.2 of \cite{MR2977317}, then we obtain, as $n\to\infty$,
\begin{eqnarray*}\lefteqn{
\Bigg[\left(\frac{1}{n}\sum_{k=1}^n\cos(k\lambda_{1})S_k(Y)\right),
\cdots,\left(\frac{1}{n}\sum_{k=1}^n\cos(k\lambda_{s})S_k(Y)\right),}\\
&& \left(\frac{1}{n}\sum_{k=1}^n\sin(k\lambda_{i})S_k(Y)\right),\cdots,\left(\frac{1}{n}\sum_{k=1}^n\sin(k\lambda_{s})S_k(Y)\right)\Bigg]\overset{d}{\rightarrow}\mathcal{N}\left((0,\delta(1/2)\Sigma(1/2)\right). 
\end{eqnarray*}

{\bf Case $\mathbf{1/2<d<3/2}$:} In this case, the differenced process $Y_k$ is an I(0) process with memory parameter $d-1\in(-1/2,1/2)$, so that 
$$\textrm{Var}(S_n(Y))\sim\delta(d-1)n^{2(d-1)+1}=\delta(d-1) n^{2d-1},
$$
where $B$ is the $\beta$ function, extended to negative numbers in $(-1,0)$. Then we get
\begin{eqnarray}\label{cumsum}\lefteqn{
    \textrm{Cov}\left[\left(\frac{1} {n^{1/2+d}}\sum_{k=1}^n\cos(k\lambda_{i})S_k(Y)\right),\left(\frac{1}{n^{1/2+d}}\sum_{k=1}^n\cos(k\lambda_{j})S_k(Y)\right)\right]}\nonumber\\
    &&=\frac{1}{n^{1+2d}}\sum_{k=1}^n\sum_{k'=1}^n\cos\left(\frac{2\pi ik}{n}\right)\cos\left(\frac{2\pi jk'}{n}\right)\mathbb{E}\left(S_k(Y)S_{k'}(Y)\right)\nonumber\\ 
    &&=\frac{1}{2}\frac{1}{n^{1+2d}}\sum_{k=1}^n\sum_{k'=1}^n\cos\left(\frac{2\pi ik}{n}\right)\cos\left(\frac{2\pi jk'}{n}\right)(-\mathbb{E}(S^2_{\vert k-k'\vert}(Y))\nonumber\\
    &&\sim\frac{\delta(d-1)}{2}\frac{1}{n^{1+2d}}\sum_{k=1}^n\sum_{k'=1}^n\cos\left(\frac{2\pi ik}{n}\right)\cos\left(\frac{2\pi jk'}{n}\right)\left(-\left(\frac{\vert k-k'\vert}{n}\right)^{2d-1}\right)n^{2d-1}\nonumber\\
&&\to-\frac{\delta(d-1)}{2}\int_{[0,1]^2]}\cos(2\pi ix)\cos(2\pi jy)(\vert x- y\vert^{2d-1})dxdy.
\end{eqnarray}
Since 
$$
\frac{1} {n^{1/2+d}}\sum_{k=1}^n\cos(k\lambda_{j})S_k(Y)=\frac{1} {n^{1/2+d}}\sum_{k=1}^n\left(\sum_{u=1}^{k-1}\cos(u\lambda_{j})\right)Y_k
$$
and that because $d>1/2$
$$
\underset{k\le n}{\sup}\vert z_{n,k}\vert:=\underset{k\le n}{\sup}\frac{1} {n^{1/2+d}}\left\vert\sum_{u=1}^{k-1}\cos(u\lambda_{j})\right\vert\to0
$$
and
$$
\sum_{k=1}^nz_{n,k}^2\le\frac{1}{n^{2d-1}}\to0,
$$
then by Theorem 4.3.2 of \cite{MR2977317} and Lemma \ref{invertible},  we obtain part (1) of the Lemma \ref{finite dim}, i.e. when  $1/2<d<3/2$.\\\\
{\bf Case $\mathbf{-1/2<d<1/2}$:} 
Before proceeding with the proof of this case, we note that \cite{MR1243575} established a smilar  result for the covariance matrix $\Sigma(d)$ (with different form of the integral) under the rather restrictive assumption that the spectral density $f$ is of the form
\[
f(\lambda)=|\lambda|^{-2d}f^*(\lambda), \qquad f^* \text{ continuous and positive},
\]
a condition on which their proof heavily relies and that we do not impose here.\\
Writing $X_k=S_k-S_{k-1}$ where, being the sum, $S_k$ is an I(1) process with memory parameter $1/2<d+1<3/2$, and using summation by parts, we obtain that 
\begin{eqnarray}\label{expansion}\lefteqn{
\sum_{k=1}^n\cos(k\lambda_{j})X_k=(\cos(\lambda_j))(S_n-X_0)-\sum_{k=1}^n\left[\cos((k+1)\lambda_j)-\cos(k\lambda_j\right]S_k}\nonumber\\
&&=(\cos(\lambda_j))S_n+
\sum_{k=1}^n\left[\cos(k\lambda j)(1-\cos(\lambda j))+\sin(k\lambda j)\sin(\lambda j)\right]S_k.
\end{eqnarray}
Hence,  applying (\ref{cumsum}) with $d+1$ instead of $d$, we obtain that, as $n\to\infty$, (since  $\cos(\lambda_j)\sim1$, $1-\cos(\lambda j)\sim(2\pi j/n)^2/2$ and $\sin(\lambda j)\sim2\pi j/n$),
\begin{eqnarray}\label{match2}\lefteqn{
    \textrm{Cov}\left[\left(\frac{1} {n^{1/2+d}}\sum_{k=1}^n\cos(k\lambda_{i})X_k\right),\left(\frac{1}{n^{1/2+d}}\sum_{k=1}^n\cos(k\lambda_{j})X_k\right)\right]}\\
&&
\to\delta(d)\Bigg[1+\frac{1}{2}\int_0^1\left(2\pi i\sin(2\pi ix)+2\pi j\sin(2\pi jx)\right)\left(x^{2d+1}-(1-x)^{2d+1}\right)dx\nonumber\\
&&-\frac{(2\pi i)(2\pi j)}{2}\int_{[0,1]^2}\sin(2\pi ix)\sin(2\pi jy)\vert x- y\vert^{2d+1}dxdy\Bigg]\nonumber\\
&&=\delta(d)\Bigg[-1+(2d+1)\int_0^1x^{2d}\left(\cos(2\pi ix)+\cos(2\pi jx)\right)dx\nonumber\\
&&-\frac{(2\pi i)(2\pi j)}{2}\int_{[0,1]^2}\sin(2\pi ix)\sin(2\pi jy)\vert x- y\vert^{2d+1}dxdy\Bigg]
\end{eqnarray}
and similarly we obtain that
\begin{eqnarray*}\lefteqn{
    \textrm{Cov}\left[\left(\frac{1} {n^{1/2+d}}\sum_{k=1}^n\sin(k\lambda_{i})X_k\right),\left(\frac{1}{n^{1/2+d}}\sum_{k=1}^n\sin(k\lambda_{j})X_k\right)\right]}\\
&&\to-\delta(d)(2\pi^2 ij)\int_{[0,1]^2}\cos(2\pi ix)\cos(2\pi jy)\vert x- y\vert^{2d+1}dxdy.
\end{eqnarray*}
The sequence of coefficients 
$$
z^{(c)}_{n,k}:=\frac{\cos(2\pi jk)}{n^{1/2+d}}\qquad\textrm{and }z^{(s)}_{n,k}:=\frac{\sin(2\pi jk)}{n^{1/2+d}},\qquad k=1,\ldots,n,\qquad j=1,\ldots,s,
$$
in the sequence of random vectors in part (2) of the Lemma \ref{finite dim} clearly 
satisfy conditions (i) and (iii) (respectively (ii)) of Proposition 4.3.1 of  \cite{MR2977317} when $0<d<1/2$ (respectively $-1/2<d<0$) and hence by Theorem 4.3.2 of the same reference, we obtain the result in part (2) of the Lemma \ref{finite dim}. This completes the proof of the Lemma.\\ 

\end{proof}

{\bf Proof of Lemma \ref{denom}} \begin{proof}
Observe that $\frac{I_{n,h}(\lambda'_j)}{(\lambda'_j)^{-1}}$
$h=1,\ldots,m$ are identically distributed, so that we have for fixed $j=1,\ldots,s$, as $n\to\infty$, 
$$
\mathbb{E}\left(\frac{1}{m}\sum_{h=1}^m\frac{I_{n,h}(\lambda'_j)}{(\lambda'_j)^{-2d}}\right)=\mathbb{E}\left(\frac{I_{n,1}(\lambda'_j)}{(\lambda'_j)^{-2d}}\right)=\left((2\pi j)^{2d}\right)\frac{I_{n,1}(\lambda'_j)}{n^{1+2d}}\to\frac{(2\pi j)^{2d}}{2\pi}I_{jj}
$$
by Lemma \ref{finite dim} and hence it will be enough to show that, as $n\to\infty$,
\begin{equation}\label{lim-var}
    \textrm{Var}\left(\frac{1}{m}\sum_{h=1}^m\frac{I_{n,h}(\lambda'_j)}{\ell^{2d}}\right)\to0.
\end{equation}
We can write
 $$
I_{n,h}(\lambda_j')=\frac{1}{2\pi \ell}\left\vert\sum_{t=1}^{\ell}X_ke^{\i t\lambda_j'}\right\vert^2=\frac{1}{2\pi \ell}\left\vert\sum_{t=1}^{\ell}S_t^{(h)}(Y)e^{\i t\lambda_j'}\right\vert^2
 $$
 where, for each $t$ such that  $1\le t\le \ell$,
 $$
 S_t^{(h)}(Y)=\sum_{u=1}^tY_{(h-1)\ell+u},\qquad h=1\ldots,m.
 $$
Let 
$$
U_{n,h}(d)=\begin{cases}
    \frac{1}{\ell^{1/2+d}}\displaystyle\sum_{k=1}^\ell\cos(k\lambda'_j)S_k^{(h)}(Y),&\textrm{if }1/2\le d<3/2,\\\\
    \frac{1}{\ell^{1/2+d}}\displaystyle\sum_{k=1}^\ell\cos(k\lambda'_j)X_k,&\textrm{if }-1/2< d<1/2,
\end{cases}
$$
and
$$
V_{n,h}(d)=\begin{cases}
    \frac{1}{\ell^{1/2+d}}\displaystyle\sum_{k=1}^\ell\sin(k\lambda'_j)S_k^{(h)}(Y),&\textrm{if }1/2\le d<3/2,\\\\
    \frac{1}{\ell^{1/2+d}}\displaystyle\sum_{k=1}^\ell\sin(k\lambda'_j) X_k,&\textrm{if }-1/2< d<1/2,
\end{cases}
$$
and let
$$
d_{\ell,u}^{(c)}=\begin{cases}
\frac{1}{\ell^{1/2+d}}\displaystyle\sum_{k=1}^\ell\left(\sum_{i=1}^ka_{i-u}\right)\cos\left(\frac{2\pi jk}{\ell}\right),&\textrm{if }1/2\le d<3/2,\\\\   
\frac{1}{\ell^{1/2+d}}\displaystyle\sum_{k=1}^\ell a_{k-u}\cos\left(\frac{2\pi jk}{\ell}\right),&\textrm{if }-1/2<d<1/2.
\end{cases}
$$
$$
d_{\ell,h,u}^{(c)}=\begin{cases}
\frac{1}{\ell^{1/2+d}}\displaystyle\sum_{k=1}^\ell\left(\sum_{i=1}^ka_{(h-1)\ell+i-u}\right)\cos\left(\frac{2\pi jk}{\ell}\right),&\textrm{if }1/2\le d<3/2,\\\\ 
\frac{1}{\ell^{1/2+d}}\displaystyle\sum_{k=1}^\ell a_{(h-1)\ell+k-u}\cos\left(\frac{2\pi jk}{\ell}\right),&\textrm{if }-1/2<d<1/2.
\end{cases}
$$
Using the stationarity of $X_k$ (when $-1/2<d<1/2$) and $Y_k=X_k-X_{k-1}$ (when $1/2\le d<3/2)$,  the left hand side of (\ref{lim-var}) can then be written as
\begin{eqnarray*}\lefteqn{
\frac{1}{m^2}\left[\sum_{i=1}^m\sum_{j=1}^m\textrm{Cov}(U_{n,i}^2(d),U_{n,j}^2(d))+\sum_{i=1}^m\sum_{j=1}^m\textrm{Cov}(V_{n,i}^2(d),V_{n,j}^2(d))+2\sum_{i=1}^m\sum_{j=1}^m\textrm{Cov}(U_{n,i}^2(d),V_{n,j}^2(d))\right]}\\
&&=\frac{1}{m}\left[\textrm{Var}(U^2_{n,1}(d))+\textrm{Var}(V_{n,1}^2)+2\textrm{Cov}(U^2_{n,1}(d),V_{n,1}^2(d))\right]+\Bigg[\frac{2}{m}\sum_{h=2}^{m-1}\left(1-\frac{h-1}{m}\right)\\
&&\left[\textrm{Cov}(U^2_{n,1}(d),U^2_{n,h}(d))+\textrm{Cov}(V_{n,1}^2(d),V^2_{n,h}(d))+\textrm{Cov}(U^2_{n,1}(d),V^2_{n,h}(d))+\textrm{Cov}(V_{n,1}^2(d),U^2_{n,h}(d))\right]\Bigg].
\end{eqnarray*}
We have
\begin{equation}\label{4th moment}
\mathbb{E}\left(U_{n,1}^4\right)=
\mathbb{E}(\epsilon_1^4)\sum_{u=-\infty}^\ell d_{\ell,u}^4+\left(\mathbb{E}(\epsilon_1^2)\sum_{u=-\infty}^\ell d_{\ell,u}^2\right)^2\le C\left(\mathbb{E}(\epsilon_1^2)\sum_{u=-\infty}^\ell d_{\ell,u}^2\right)^2\to C\Sigma^{(c)}_{jj}(d)<\infty,
\end{equation}
where $C$ is a generic positive constant,
and similarly  we can show that $\mathbb{E}\left(V_{n,1}^4\right)$ is bounded. It will then be enough to show that, uniformly in $h=1,\ldots,m$, as $n\to\infty$,
\begin{equation*}
\left\vert\textrm{Cov}(U^2_{n,1}(d),U_{n,h}^2(d)\right\vert+\left\vert\textrm{Cov}(U^2_{n,1}(d),V_{n,h}^2(d))\right\vert+\left\vert\textrm{Cov}(V_{n,1}^2(d),U_{n,h}^2(d))\right\vert+\left\vert\textrm{Cov}(V_{n,1}^2(d),V_{n,h}^2(d))\right\vert\to0.
\end{equation*}
We will show that, uniformly in $h$,
\begin{equation}\label{cov-lim0}
    \textrm{Cov}(U^2_{n,1}(d),U_{n,h}^2(d)\to0,\qquad\textrm{as }n\to\infty.
\end{equation}
The other covarainces treat similarly.\\
{\bf Proof of (\ref{cov-lim0}) when $\mathbf{1/2\le d<3/2:}$}
We have for $h\ge2$, 
$$
\textrm{Cov}(U^2_{n,1}(d),U_{n,h}^2(d))=\mathbb{E}(U^2_{n,1}(d)U^2_{n,h}(d))-\left(\mathbb{E}(U^2_{n,1}(d))\right)^2,
$$
\begin{eqnarray}\label{cov-cov}
\mathbb{E}(U^2_{n,1}(d)U_{n,h}^2(d))&=&\mathbb{E}\left[\left(\sum_{u=-\infty}^\ell d_{\ell,u}\epsilon_u\right)^2\left(\sum_{u=-\infty}^{h\ell} d_{\ell,h,u}\epsilon_u\right)^2\right]\\
&=&\mathbb{E}\left[\left(\sum_{u=-\infty}^\ell d_{\ell,u}\epsilon_u\right)^2\left(\sum_{u=-\infty}^{(h-1)\ell} d_{\ell,h,u}\epsilon_u\right)^2\right]+
\left(\mathbb{E} (U^2_{n,1}(d))\right)\sigma^2_\epsilon\sum_{u=(h-1)\ell+1}^{h\ell}d_{\ell,h,u}^2,\nonumber\\
\end{eqnarray}
and clearly 
$$
\sum_{u=(h-1)\ell+1}^{h\ell}d_{\ell,h,u}^2=\sum_{u=0}^\ell d_{\ell,u}^2.
$$
We first show that
\begin{equation}\label{l limit}
\sum_{u=-\infty}^0d_{\ell,u}^2\to0,
\end{equation}
and
\begin{equation}\label{h limit}
\sum_{u=-\infty}^{(h-1)\ell} d_{\ell,h,u}^2\to0\qquad\textrm{uniformly in  }h,\qquad\textrm{as }\ell\to\infty.
\end{equation}
For $d\neq1$, and as $\ell\to\infty$,
\begin{eqnarray*}\lefteqn{
\sum_{u=-\infty}^0d_{\ell,u}^2\le C\ell^{-2d}\sum_{u=0}^\infty\left((1+u)^{d-1}-(1+\ell+u)^{d-1}\right)^2}\\
&&=C\ell^{-2d}\sum_{u=0}^\ell\left((1+u)^{d-1}-(1+\ell+u)^{d-1}\right)^2+C\ell^{-2d}\sum_{u=\ell+1}^\infty\left((1+u)^{d-1}-(1+\ell+u)^{d-1}\right)^2\\
&&\to0,
\end{eqnarray*}
since the first sum in the line above is equivalent to
$$
\frac{1}{\ell}\int_0^1\left(x^{d-1}-(1+x)^{d-1}\right)^2dx\to0,
$$
and the second sum is $O(\ell^{-1})$ by Mean-Value Theorem.\\
When $d=1$ then $a_{i}$ is summable and we get
\begin{eqnarray*}
\sum_{u=-\infty}^0d_{\ell,u}^2&\le& \frac{1}{\ell}\sum_{u=0}^\infty
\left(\sum_{i=1}^\ell a_{i+u}\right)^2\le\frac{C}{\ell}\sum_{i=1}^\ell\sum_{u=0}^\infty a_{i+u}\to0.
\end{eqnarray*} 
Similarly we can show that, as $\ell\to\infty$, for each $h\ge2$,
$$
\sum_{u=-\infty}^0d_{\ell,h,u}^2\to0.
$$
Also, for $d\neq1$,
\begin{eqnarray*}
\sum_{u=0}^{(h-1)\ell} d_{\ell,h,u}^2&\le& \frac{1}{\ell^{2d}}\sum_{u=0}^{(h-1)\ell}
\left(\sum_{i=1}^\ell a_{(h-1)\ell+i-u}\right)^2\\
&\le&\frac{1}{\ell^{2d}}\sum_{u=1}^{(h-1)\ell}\left(1+(h-1)\ell-u)^{d-1}-(1+h\ell-u)^{d-1}\right)^2\\
&\le&\begin{cases} 
    \frac{1}{\ell^{2d}}\displaystyle\sum_{u=0}^{(h-1)\ell}[1+(h-1)\ell-u]^{d-2}\ell=O\left(\ell^{1-2d}\right)&\textrm{if }1/2<d<1,\\\\
  \frac{1}{\ell^{2d}}\displaystyle\sum_{u=0}^{(h-1)\ell}[1+(h-1)\ell-u]^{2d-4}\ell^2=O\left(\ell^{2-2d}\right),&\textrm{if }1<d<3/2.
\end{cases}
\end{eqnarray*} 
When $d=1$, then $a_i$ are summable and therefore as $\ell\to\infty$,
$$
\sum_{u=0}^{(h-1)\ell} d_{\ell,h,u}^2\le\frac{1}{\ell}\sum_{i=1}^\ell\sum_{u=0}^\infty a_{i+u}\to0.
$$
So, we get
\begin{equation}\label{limh h}
\sum_{u=-\infty}^{(h-1)\ell}d_{\ell,h,u}^2\to0,\qquad\textrm{as }\ell\to\infty.
\end{equation}
This completes the proofs of (\ref{l limit}) and (\ref{h limit})  in the case $1/2\le d<3/2$.
Now, using (\ref{4th moment}),
\begin{eqnarray*}\lefteqn{
\left(\mathbb{E}\left[\left(\sum_{u=-\infty}^\ell d_{\ell,u}^{(c)}\epsilon_u\right)^2\left(\sum_{u=-\infty}^{(h-1)\ell} d_{\ell,h,u}^{(c)}\epsilon_u\right)^2\right]\right)^2}\\
&&\le
\mathbb{E}\left[\left(\sum_{u=-\infty}^\ell d_{\ell,u}^{(c)}\epsilon_u\right)^4\right]\mathbb{E}\left[\left(\sum_{u=-\infty}^{(h-1)\ell} d_{\ell,h,u}^{(c)}\epsilon_u\right)^4\right]\\
&&\le C\left(\sum_{u=-\infty}^\ell d_{\ell,u}^2\right)^2\left(\sum_{u=-\infty}^{(h-1)\ell}d_{\ell,h,u}^2\right)^2\to0
\end{eqnarray*}
by (\ref{limh h}). Therefore, using (\ref{cov-cov}), we conclude  the proof of (\ref{cov-lim0}) when $1/2\le d<3/2$. The Same proof can be used  to show that
\begin{equation}\label{vn-vn}
\textrm{Cov}(V^2_{n,1}(d),V_{n,h}^2(d))\to0,\qquad\textrm{as }n\to\infty
\end{equation}
holds true when $1/2\le d<3/2.$\\\\
{\bf Proof of (\ref{cov-lim0}) when $\mathbf{-1/2<d<1/2}$:}  Using the representation (\ref{expansion}) with $\ell$ instead of $n$ and $\lambda'_j$ instead of $\lambda_j$, we obtain that for $h\ge1$,
$$
U_{n,h}(d)=\cos(\lambda'_j)(S_\ell^{(h)}-X_{h-1})+\ell(1-\cos(\lambda'_j))U_{n,h}(d+1)+\ell\sin(\lambda'_j)V_{n,h}(d+1),
$$
and
$$
V_{n,h}(d)=\sin(\lambda'_j)(S_\ell^{(h)}-X_{h-1})-\ell(1-\cos(\lambda'_j))V_{n,h}(d+1)-\ell\sin(\lambda'_j)U_{n,h}(d+1),
$$
and therefore we obtain that, as $\ell\to\infty$, uniformly in $h\ge2$,
\begin{eqnarray}\label{final cov}
\textrm{Cov}\left(U^2_{n,1}(d),U_{n,h}^2(d)\right)
&\sim&\textrm{Cov}\left(\frac{S^2_\ell} {n^{2d+1}},\frac{\left(S_\ell^{(h)}(X)\right)^2}{n^{2d+1}}\right)+
(2\pi j)^2\textrm{Cov}\left(V^2_{n,1}(d+1),V^2_{n,h}+
(d+1)\right)\nonumber\\
&&+(2\pi j)^2\textrm{Cov}\left(\frac{S_\ell}{n^d+1/2}V_{n,1}(d+1),\frac{S_\ell^{(h)}(X)}{n^{d+1/2}}V_{n,h}(d+1)\right).
\end{eqnarray}
The middle covariance in the right hand side is already shown to converge to zero as $n\to\infty$ as per (\ref{vn-vn}). The first covariance can also be shown to converge to zero in a similar way (as the proof does not rely on $\lambda_j$. For the last one, similar decomposition as in (\ref{cov-cov}) can be used to show that, uniformly in $h\ge2$,
\begin{eqnarray*}\lefteqn{
\mathbb{E}\left(\frac{S_\ell}{\ell^{d+1/2}}V_{n,1}(d+1)\frac{S_\ell^{(h)}(X)}{\ell^{d+1/2}}V_{n,h}(d+1)\right)}\\
&&-\left[\mathbb{E}\left(\frac{S_\ell}{\ell^{d+1/2}}V_{n,1}(d+1)\right)\right]\sigma^2_\epsilon\sum_{u=(h-1)\ell}^{h\ell}\sum_{v=(h-1)\ell}^{h\ell}d^{(s)}_{n,h,u}d'_{n,h,u}\to0,\nonumber
\end{eqnarray*}
where
$$
d^{(s)}_{n,h,u}=\frac{1}{\ell^{1/2+d}}\displaystyle\sum_{k=1}^\ell\left(\sum_{i=1}^ka_{(h-1)\ell+i-u}\right)\sin\left(\frac{2\pi jk}{\ell}\right)
$$
and
$$
d'_{n,h,u}=\sum_{i=1}^\ell(\ell-i)a_{(h-1)\ell+i-u},
$$
and that
$$
\sigma^2_\epsilon\sum_{u=(h-1)\ell}^{h\ell}\sum_{v=(h-1)\ell}^{h\ell}d^{(s)}_{n,h,u}d'_{n,h,u}\to\underset{n\to\infty}{\lim}\mathbb{E}\left(\frac{S_\ell}{n^d+1/2}V_{n,1}(d+1)\right)=c(d)\sigma^2_\epsilon\int_0^1x^{2d+1}\sin(2\pi jx)dx,
$$
so that the last covariance in (\ref{final cov}) tends to zero as $n\to\infty$, uniformly in $h\ge2$.
\end{proof}
\bibliographystyle{apalike}
\bibliography{jtsa}
\end{document}